\documentclass[final,5p,times,twocolumn]{elsarticle}

\usepackage{amssymb}
\usepackage{amsmath}
\usepackage{amsthm}
\usepackage{framed} 
\usepackage{multicol} 
\usepackage{tikz}
\usetikzlibrary{arrows.meta, positioning}
\usepackage{xcolor}
\usepackage{graphicx}

\usepackage{graphics}
\usepackage{enumitem}
\newtheorem{theorem}{Theorem}
\usepackage{subcaption}
\usepackage{booktabs}
\newtheorem{assumption}{Assumption}
\newtheorem{corollary}{Corollary} 
\newtheorem{lemma}{Lemma}
\usepackage{framed} 
\usepackage{multicol} 

\usepackage[hidelinks]{hyperref}

\begin{document}

		\makeatletter
	\def\ps@pprintTitle{}
	\makeatother

\begin{frontmatter}



\title{Fault-Tolerant Control of Steam Temperature in HRSG Superheater under Actuator Fault Using a Sliding Mode Observer and PINN} 


\author[inst1]{Mojtaba Fanoodi}

\author[inst1]{Farzaneh Abdollahi}

\author[inst2]{Mahdi Aliyari Shoorehdeli}

\affiliation[inst1]{
	organization={Department of Electrical Engineering, AmirKabir University of Technology (Tehran Polythechnique)},
	city={Tehran},
	country={Iran}
}

\affiliation[inst2]{
	organization={Department of Electrical Engineering, K.N. Toosi University of Technology},
	city={Tehran},
	country={Iran}
}


\begin{abstract}
	\noindent
	This paper presents a novel fault-tolerant control framework for steam temperature regulation in Heat Recovery Steam Generators (HRSGs) subject to actuator faults. Addressing the critical challenge of valve degradation in superheater spray attemperators, we propose a synergistic architecture comprising three components: (1) a Sliding Mode Observer (SMO) for estimation of unmeasured thermal states, (2) a Physics-Informed Neural Network (PINN) for estimating multiplicative actuator faults using physical laws as constraints, and (3) a one-sided Sliding Mode Controller (SMC) that adapts to the estimated faults while minimizing excessive actuation.
	
	The key innovation lies in the framework of closed-loop physics-awareness, where the PINN continuously informs both the observer and controller about fault severity while preserving thermodynamic consistency. 
	
	Rigorous uniform ultimate boundedness (UUB) is established via Lyapunov analysis under practical assumptions. Validated on real HRSG operational data, the framework demonstrates effective fault adaptation, reduced temperature overshoot, and maintains steam temperature within ±1°C of the setpoint under valve effectiveness loss.
	
	This work bridges control theory and physics-guided machine learning to deliver a practically deployable solution for power plant resilience, with extensions applicable to thermal systems subject to multiplicative faults.
\end{abstract}

%
%


\end{frontmatter}




\section{Introduction}

\label{sec:introduction}
The growing demand for flexible operation of combined cycle power plants has intensified challenges in maintaining critical process variables under equipment degradation. Heat recovery steam generators, which harness waste heat from gas turbines, require precise steam temperature control to ensure efficiency, prevent material stress, and avoid catastrophic tube failures. Superheater sections—where the final steam temperature is controlled—are particularly susceptible to actuator faults in spray water valves. This vulnerability that is further exacerbated by the frequent load cycling inherent in modern power systems. Fig.~\ref{hrsg} illustrates a typical heat recovery steam generator system in a combined cycle configuration.

One prevalent issue encountered within HRSG systems is the failure of superheater components. These failures often arise due to severe operational conditions, including overheating and thermal cycling, which can lead to degradation of material properties \cite{latif2023failure}. Such failures not only reduce the efficiency of the HRSG but also result in significant downtime and repair costs \cite{sanaye2007transient}. Accurate modeling is essential for understanding how actuator failures propagate through the system, affecting heat transfer processes and overall system performance \cite{self2018effects}. Simulation-based studies have enabled analysis of HRSG behavior under non-standard operating conditions, providing valuable insight for improving design robustness \cite{mcconnell2023modeling}. Furthermore, data analytics and monitoring technologies now support real-time diagnostics of actuator health \cite{sun2014study, mcconnell2019multi}, enhancing fault prediction and maintenance scheduling. These insights underscore the need for actuator fault modeling in desuperheaters, particularly multiplicative faults arising from erosion, stiction, and scaling \cite{lv2019dependency, zhang2020distributed}. 

Numerous field investigations have examined failure mechanisms in superheater and attemperator components under prolonged thermal loading. Attemperators installed in superheater steam lines have been reported to be highly prone to thermal-fatigue cracking and creep-fatigue interaction, wherein microstructural analysis revealed transgranular and intergranular cracks leading to eventual rupture~\cite{mukhopadhyay2011failure}. 
Similarly, recurrent superheater tube failures have been attributed to long-term overheating caused by attemperator over-spraying, with remedial retrofitting implemented to restore tube integrity~\cite{shokouhmand2015failure}. 
A study on a 600~MW secondary superheater header showed that localized overheating due to magnetite deposition obstructed steam flow and induced tube leakage after extended service life~\cite{nurbanasari2023failure}. 
Moreover, analysis of a T22-steel superheater tube rupture revealed ``fish-mouth'' deformation and carbide depletion, consistent with short-time overheating from burner or steam flow irregularities~\cite{kochmanski2024failure}. 

Collectively, these findings demonstrate that even small deviations in spray-water regulation or thermal transients can initiate severe mechanical degradation in HRSG superheater assemblies. 
Such evidence underscores the necessity for robust control and fault-tolerant strategies to mitigate overheating and ensure long-term reliability.

\begin{figure}
	\centering
	\includegraphics[width=0.8\columnwidth]{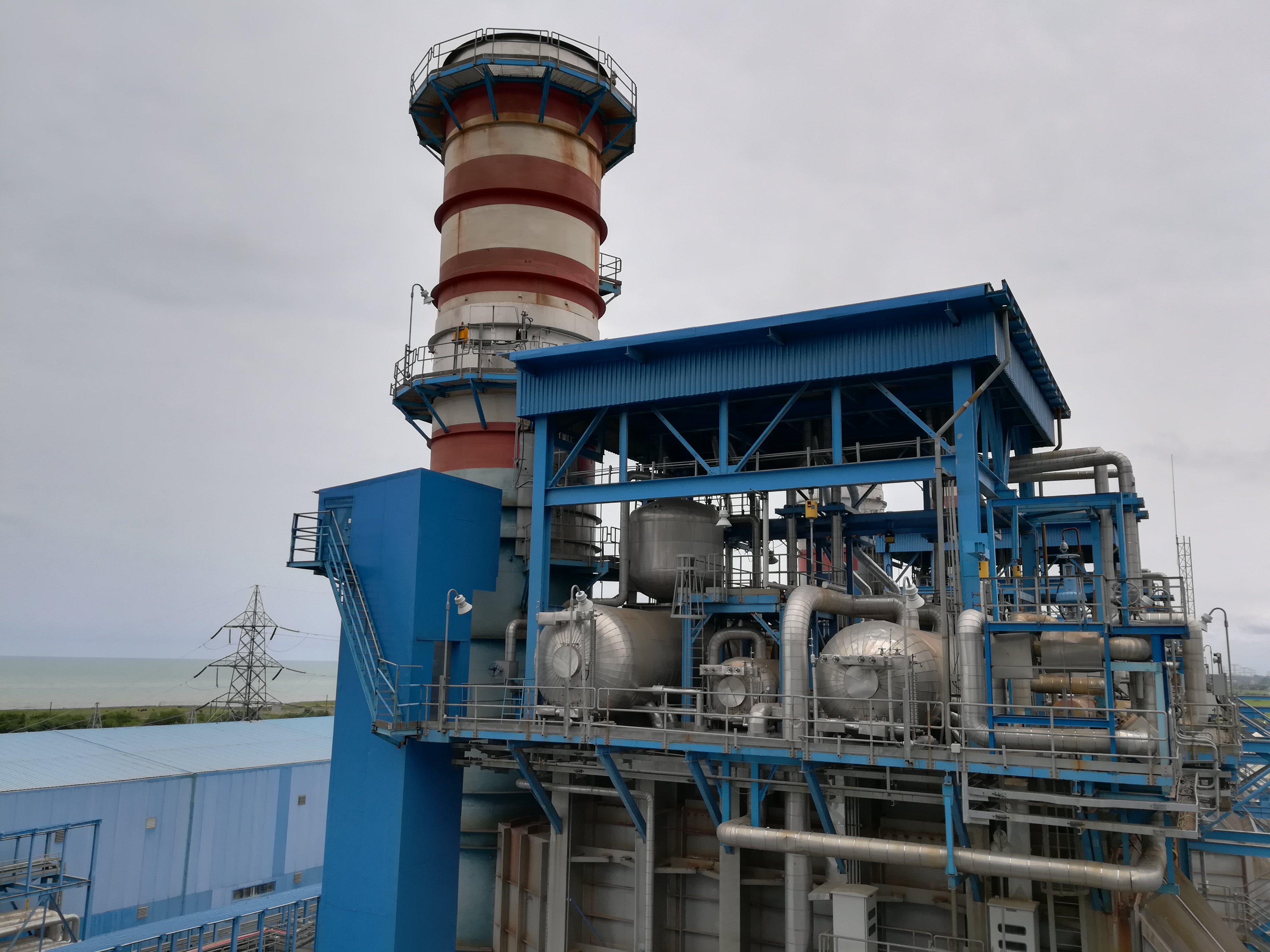}
	\caption{HRSG in Pareh-Sar Powerplant}
	\label{hrsg}
\end{figure}

Actuator faults in desuperheater valves, typically manifesting as \emph{loss of effectiveness} due to stiction, erosion, or scaling, compromise temperature control, and pose significant operational risks. These faults introduce multiplicative effects that conventional controllers cannot address. Fig.~\ref{f1} illustrates the degradation in spray performance due to loss-of-effectiveness faults in desuperheater nozzles.

\begin{figure}[h!]
	\centering
	\includegraphics[width=0.8\columnwidth]{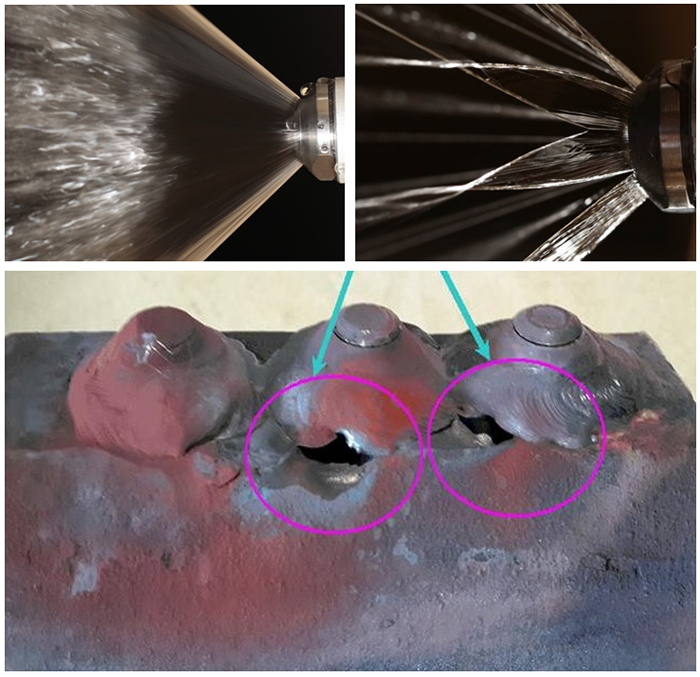}
	\caption{A fully functional desuperheater nozzle produces a finely atomized spray into the steam flow (top left). A loss-of-effectiveness fault, caused by blockage or mechanical damage, reduces spray quality (top right) or allows undispersed water to enter the steam pipe (bottom), leading to thermal shock and water hammer damage\cite{gibbons_how_2025}.}
	\label{f1}
\end{figure}

Maintenance reports from the Pareh-sar combined cycle power plant, located in northern Iran, confirm frequent occurrences of such faults. Specifically, logs show instances where the Distributed Control System (DCS) commanded a 70\% valve opening, but the valve only responded with 40\% due to degraded actuator effectiveness. Conversely, in shutdown scenarios, the controller demanded full closure, yet the valve remained partially open at around 15\%. These discrepancies resulted in steam temperatures deviating from the desired setpoints, triggering operator interventions. In several cases, operators manually adjusted the control loop setpoint to restore safe conditions. Such manual overrides not only reduced thermal efficiency in the HRSG but also negatively impacted the overall performance of the combined cycle plant. Fig.~\ref{f2} shows the types of equipment damage typically caused by prolonged desuperheater malfunction.

\begin{figure}[h!]
	\centering
	\includegraphics[width=0.8\columnwidth]{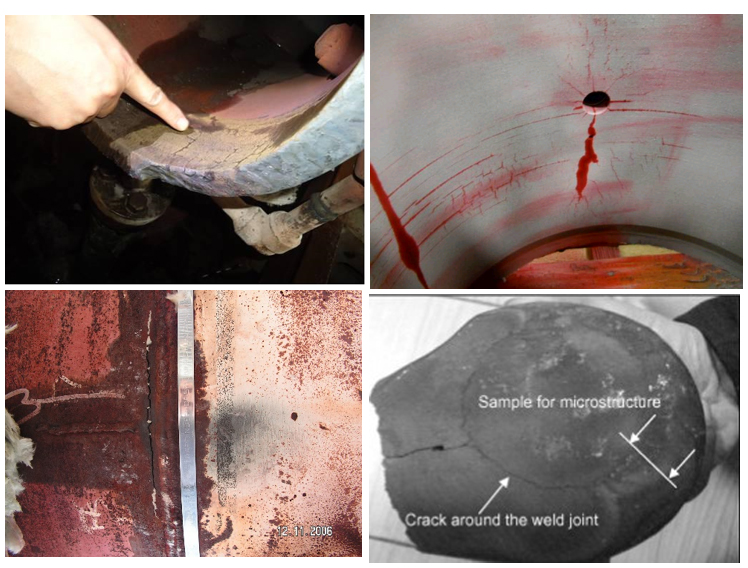}
	\caption{A loss-of-effectiveness actuator fault in the desuperheater can lead to over-temperature conditions, water hammer, or thermal shock, resulting in severe damage to piping and equipment and potentially causing extended outages for repair\cite{gibbons_how_2025,mukhopadhyay2011failure}.}
	\label{f2}
\end{figure}

Furthermore, the core challenges are compounded by (1) unmeasured states, as critical temperatures (e.g., attemperator and intermediate steam states) lack physical sensors; (2) strong nonlinearities, due to convective heat transfer dynamics with bilinear terms; and (3) real-time constraints, given the need for online fault adaptation without process interruption.

In terms of state estimation, recent advancements include distributed estimation strategies for complex systems\cite{abdin2022state, inyang2021health}. For example, Zhang et al. developed a distributed observer capable of fusing subsystem models to reconstruct internal states under fault conditions \cite{zhang2024sliding}. Similarly, Luenberger-style observers have demonstrated real-time capabilities, allowing effective tracking of internal variables even in the presence of unknown inputs\cite{ tang2023data}. 

Nonlinearities in HRSG control are another limiting factor. Conventional PI/PID control strategies often fall short due to long time delays and thermal inertia. Alternatives such as Virtual Reference Feedback Tuning (VRFT) have shown promise in thermal power unit regulation by mitigating delay-related issues \cite{zhou2023virtual}. Advanced feedback mechanisms, including backstepping control and nonlinear model-based strategies, have also been proposed for improving control fidelity in steam temperature regulation \cite{li2017superheat}. In this context, low-order ADRC schemes with phase-compensation have shown strong potential for handling the large inertia and high-order dynamics of superheated steam temperature systems, demonstrating improved robustness and disturbance rejection in practical thermal power units \cite{chen2022phase}. Similarly, modified ADRC formulations have been developed to compensate for sluggish dynamics in SST processes, offering improved tracking and disturbance rejection along with field-validated performance in large power plants \cite{wu2019superheated}. Beyond ADRC-based methods, nonlinear decoupling strategies have also been proposed for steam power plants to independently regulate pressure and steam temperature, effectively eliminating inverse-response behavior and improving robustness under large load variations \cite{alamoodi2017nonlinear}.

Regarding real-time adaptation, robust FTC schemes like fixed-time sliding-mode controllers have been proposed for waste heat systems to ensure stability despite rapid operational disturbances \cite{wang2018fixed}. Dynamic modeling approaches for combined cycle plants also demonstrate the value of real-time fault-resilient control \cite{wang2019evaluation}. Nevertheless, complex interactions—such as the pinch point phenomenon in heat exchangers—introduce sharp performance drops in thermal systems that challenge both model accuracy and fault isolation \cite{vcehil2017novel}.

Existing model-based fault-tolerant approaches struggle with unmodeled dynamics, while data-driven methods often violate physical constraints. Current solutions exhibit critical limitations: (1) While Sliding Mode Observers have been successfully applied for fault diagnosis and fault-tolerant control \cite{zhang2024sliding}, they generally suffer from limitations in fault quantification, particularly under varying operational modes where observer performance may vary and fault signatures may not be unique \cite{ guo2023nonsingular}; (2) Traditional Neural Networks estimate faults but often yield physically inconsistent results due to lacking physics-informed constraints, which are essential for improving model transparency and alignment with system dynamics \cite{li2023physics,retzler2024learning}; and (3) Adaptive Controllers require persistent excitation and struggle with transient faults, as highlighted by Jenkins et al., who expressed that insufficient excitation conditions can undermine stability \cite{jenkins2018convergence}. Recent advances in physics-guided learning further demonstrate that hybrid architectures combining physics-based models with neural networks can improve extrapolation and provide stability guarantees, as shown by new PGNN feedforward controllers with provable input-to-state stability properties \cite{bolderman2024physics}. Such approaches highlight the importance of embedding physical priors into learning-based control frameworks, especially for safety-critical thermal systems.

\begin{figure}[h!]
	\centering
	\includegraphics[width=0.8\columnwidth]{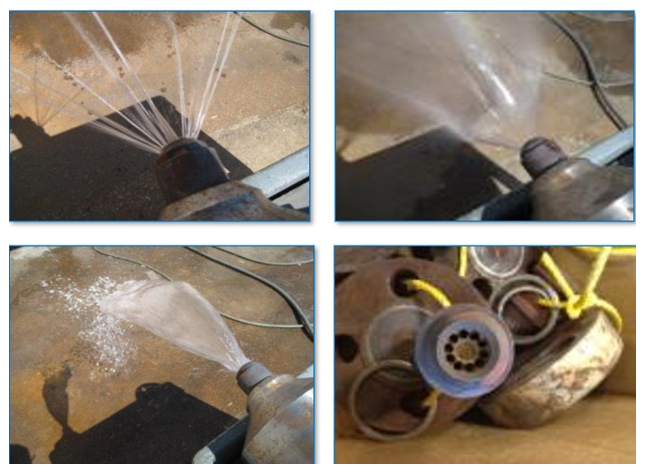}
	\caption{A loss-of-effectiveness actuator fault can manifest through various forms of nozzle damage, each degrading spray performance. Examples include plugging (top left), out-of-spec spray patterns (top right), broken nozzles (lower left), and completely detached or blown-off nozzle tips (lower right) \cite{gibbons_how_2025}.}
	\label{f3}
\end{figure}

Furthermore, purely data-driven methods such as MPC often suffer from inadequate data availability. They are prone to overfitting in the presence of non-Gaussian noise or exogenous disturbances typical in HRSGs \cite{van2020noisy}. High computational burdens also pose a challenge in AGC contexts where quick response times are essential \cite{sabounchi2021fltrl}. Finally, real-world deployment of redundant fault detection mechanisms increases complexity and cost, reducing compatibility with existing industrial infrastructure \cite{dijoux2022experimental, adumene2015performance}. Human factors also remain critical, as operators must often manually intervene, reducing trust in automated fault recovery and increasing reliance on experiential heuristics \cite{kaviri2012modeling}. These challenges are often rooted in the unpredictable mechanical behavior of degraded hardware. Fig.~\ref{f3} provides visual examples of physical nozzle damage that contribute to loss-of-effectiveness actuator faults in desuperheater systems.

To bridge these gaps, we propose a novel hybrid architecture integrating (1) a Physics-Informed Neural Network for estimation of multiplicative actuator faults using system dynamics as constraints; (2) a sliding mode observer for reconstruction of unmeasured states under fault conditions; and (3) a one-sided sliding mode control strategy for fault-tolerant temperature regulation with minimal actuation. This synergistic integration creates a physics-aware fault-tolerant control loop where the PINN informs both observer and controller about fault severity.

The paper's primary contributions are:
\begin{itemize}
	\item Employing PINNs for multiplicative actuator fault estimation in thermal energy systems.
	\item Providing a stable integration of SMO-PINN-SMC with Lyapunov-based uniform ultimate boundedness guarantees.
	\item A one-sided control law is implemented by minimizing valve actuation while preventing overtemperature.
\end{itemize} 

Experimental validation using real HRSG operational data demonstrates fault adaptation and smaller temperature overshoots compared to conventional PID controllers currently used in the industry.

The remainder of this paper is structured as follows.  
The nonlinear HRSG system, including the superheater-desuperheater dynamics and fault modeling, is introduced in Section~\ref{sec:model}.  
A detailed account of the SMO-PINN-SMC co-design methodology—covering fault estimation, state reconstruction, and control synthesis—is provided in Section~\ref{sec:method}.  
Stability analysis, based on Lyapunov theory and focused on uniform ultimate boundedness under actuator degradation, is presented in Section~\ref{sec:stability}.  
Performance validation using industrial HRSG data is reported in Section~\ref{sec:results}, demonstrating improved fault adaptation, reduced overshoot, and lower actuation demand.  
Concluding remarks and potential directions for future research are discussed in the final section, Section~\ref{sec:conclusion}.

\section{The System Model} \label{sec:model}

Accurate modeling of thermal subsystems in heat recovery steam generators is essential for effective control, particularly under actuator degradation. Due to the inherently nonlinear nature of convective heat transfer and mass flow interactions—especially in the superheater and desuperheater regions—linear approximations often fail to capture the system's transient and steady-state behavior under varying loads. Additionally, bilinear coupling between control inputs and temperature differences necessitates a nonlinear representation for fault diagnosis and robust control synthesis.

We consider a nonlinear model of a subsystem inside an HRSG, focusing on desuperheater/superheater temperature control. The state-space representation is given as~\cite{fanoodi2025pinn}:
\begin{equation}
	\begin{cases} 
		\dot{x}_1 = K_2\left(K_1 d_1 + d_7(x_2 - x_1) - d_3\right), \\ 
		\dot{x}_2 = K_3\left[(d_2 + u)(d_4 - x_2) -  u(d_4 - d_5) + \bar{m}_{\text{in}} d_6\right], \\
		y = x_1 
	\end{cases} \label{System_nominal}
\end{equation}

where:
\begin{itemize}
	\item \( x_1 \): outlet steam temperature (measured),
	\item \( x_2 \): intermediate attemperator steam temperature (unmeasured),
	\item \( u \): control input (spray valve actuation),
	\item \( y = x_1 \): system output.
\end{itemize}

\begin{table}[h]
	\centering
	\caption{Definition of Disturbance Terms in the HRSG Model}
	\resizebox{\columnwidth}{!}{
	\begin{tabular}{cll}
		\toprule
		\textbf{Symbol} & \textbf{Expression} & \textbf{Physical Meaning} \\
		\midrule
		$d_1$ & $\dot{m}_f$ & Fuel flow rate \\
		$d_2$ & $\dot{m}_{\text{in}_{\text{dsh}}}$ & Inlet mass flow rate to desuperheater \\
		$d_3$ & $\dot{\tilde{T}}_a$ & Derivative of the metallic parts temperature \\
		$d_4$ & $T_{\text{in}_{\text{dsh}}}$ & Desuperheater inlet steam temperature \\
		$d_5$ & $T_{\text{spray}}$ & Spray water temperature \\
		$d_6$ & $\dot{T}_{\text{in}_{\text{dsh}}}$ & Derivative of inlet steam temperature \\
		$d_7$ & $\dot{m}_{\text{out}_{\text{dsh}}}$ & Outlet mass flow rate from desuperheater \\
		\bottomrule
	\end{tabular}}
	\label{tab:disturbances}
	\end{table}

The model parameters are defined as:
\[
K_1 = \frac{H}{C_p}, \quad K_2 = \frac{1}{\rho_s V_s}, \quad K_3 = \frac{1}{\bar{m}_{\text{out}_{\text{dsh}}}}
\]

To account for actuator degradation, we introduce a \emph{multiplicative loss of effectiveness fault} into the spray valve actuation. The faulty system dynamics are represented as:

\begin{equation}
	\begin{cases} 
		\dot{x}_1 = K_2\left(K_1 d_1 + d_7(x_2 - x_1) - d_3\right), \\ 
		\dot{x}_2 = K_3\left[(d_2 + (1-\phi) u)(d_4 - x_2) - (1-\phi) u(d_4 - d_5) + \bar{m}_{\text{in}} d_6\right], \\
		y = x_1 
	\end{cases} \label{System_faulty}
\end{equation}

Here, \( \phi(t) \in [0, 1] \) denotes the time-varying actuator fault level, with:
\[
\text{Effective input} = (1 - \phi(t)) u
\]

\begin{itemize}
	\item \( \phi = 0 \): healthy actuator (full effectiveness),
	\item \( \phi < 1 \): degraded actuator.
\end{itemize}

This formulation allows the fault to directly scale the control input, preserving the multiplicative structure observed in physical degradation mechanisms such as stiction, scaling, or erosion. It forms the basis for the fault estimation and fault-tolerant control strategies developed in the following sections.

The objective is to design a fault-tolerant controller capable of maintaining temperature regulation in the HRSG superheater subsystem despite actuator degradation. By explicitly accounting for the multiplicative fault structure in (\ref{System_faulty}), the proposed control strategy will (1) estimate the time-varying fault level $\hat{\phi}(t)$, (2) compensate for the loss of actuation effectiveness through sliding mode control reconfiguration, and (3) preserve closed-loop stability and performance under both nominal and faulty conditions. The nonlinear model dynamics, coupled with the fault parametrization, provide a foundation for synthesizing a controller that addresses the inherent challenges—ensuring operational safety and efficiency during prolonged actuator degradation.

\section{Sliding Mode Observer-Physics Informed Neural Network-Sliding Mode Controller Architecture} \label{sec:method}

In practical HRSG configurations, not all internal states are directly measurable. While the outlet steam temperature \( x_1 \) is typically monitored, the intermediate steam temperature \( x_2 \) can be unmeasured due to sensor limitations or hardware cost constraints. Accurate knowledge of \( x_2 \) is, however, essential for estimating actuator effectiveness and ensuring reliable control.

To implement model-based fault-tolerant control and estimate the unknown fault signal \( \phi(t) \), the control architecture requires:
\begin{itemize}
	\item Estimation of all system states, especially \( x_2 \),
	\item Residual information capable of inferring \( \phi(t) \).
\end{itemize}

To address these needs, we propose a hybrid estimation structure consisting of:
\begin{itemize}
	\item A Sliding Mode Observer (SMO) to reconstruct unmeasured states using available measurements and control inputs,
	\item A Physics-Informed Neural Network (PINN) to estimate the multiplicative actuator fault \( \hat{\phi}(t) \) using system data and physics-based constraints.
\end{itemize}

\subsection{Sliding Mode Observer (SMO)}

Assuming the control input \( u(t) \) is computed from the fault tolerant controller and the PINN provides an estimate \( \hat{\phi}(t) \), the SMO is formulated as:

\begin{equation}
	\resizebox{\columnwidth}{!}{$
		\begin{cases}
			\begin{aligned}
				\dot{\hat{x}}_1 &= K_2 \left( K_1 d_1 + d_7 (\hat{x}_2 - \hat{x}_1) - d_3 \right) + \lambda_1 \cdot \text{sat}\left( \frac{e_1}{\delta_1} \right) \\
				\dot{\hat{x}}_{2}&=K_{3}\left[\left(d_{2}+(1-\hat{\phi})u\right)(d_{4}-\hat{x}_{2})-(1-\hat{\phi})u(d_{4}-d_{5})+\bar{m}_{\rm in}d_{6}\right]+ \lambda_2 \,\mathrm{sat}\!\left(\tfrac{e_2}{\delta_2}\right) 
			\end{aligned}
		\end{cases} \label{SMO} $}
\end{equation}

where:
\begin{itemize}
	\item \( e_1 = x_1 - \hat{x}_1 \) is the measurable output error,
	\item \( e_2 = x_2 - \hat{x}_2 \) is the unmeasured state error,
	\item \( \lambda_1, \lambda_2 \): observer sliding gains,
	\item \( \delta_1, \delta_2 \): boundary layer thicknesses.
\end{itemize}

Here the sliding surface is implicitly defined as \( e_1 = 0 \). Sliding mode is enforced through the corrective term \( \lambda_1 \cdot \text{sat}(e_1/\delta_1) \), which drives \( e_1 \) into a boundary layer of thickness \( \delta_1 \), approximating ideal convergence (\( e_1 \to 0 \)) while minimizing chattering. The unmeasured state \( \hat{x}_2 \) is reconstructed via the system-coupled dynamics and the secondary error injection term \( \lambda_2 \cdot \text{sat}(e_1/\delta_2) \), exploiting the inherent observability of the system. This sliding-mode formulation achieves state estimation under disturbances and actuator degradation, enabling fault-tolerant control without explicit surface parameterization.This observer leverages the known structure of the plant model to drive the estimation error to a small neighborhood of zero.

\subsection{Physics-Informed Neural Network (PINN)}

Actuator faults in HRSG systems are modeled as time-varying multiplicative losses of effectiveness \( \phi(t) \in [0, 1] \). When \( \phi(t) < 1 \), the system experiences reduced actuation due to degradation phenomena such as stiction or partial valve blockage. Since \( \phi(t) \) is unmeasured and cannot be observed directly, we employ a PINN to estimate it in real time using system data and structural dynamics.

  At each time step, the input vector consists of:
\begin{itemize}
	\item Measured outlet steam temperature \( x_1(t) \),
	\item Estimated intermediate temperature \( \hat{x}_2(t) \) from the SMO,
	\item Control input \( u(t) \).
\end{itemize}
The output is the estimated fault signal:
\begin{equation}
	\hat{\phi}(t) = \text{PINN}(x_1(t), \hat{x}_2(t), u(t); W), \label{PINN}
\end{equation}
where \( W \) denotes the neural network parameters as weights and biases.

The PINN is trained by minimizing a loss function that incorporates the nonlinear plant dynamics of \( x_2 \). Specifically, the physics-based term penalizes deviation from the predicted time derivative of \( x_2 \):

\begin{equation}
		\resizebox{\columnwidth}{!}{$
		\mathcal{L}_{\text{physics}} = \left\| 
		\frac{d\hat{x}_2}{dt} - 
		K_3 \left[ 
		\left(d_2 + (1-\hat{\phi})u \right)(d_4 - \hat{x}_2) 
		- (1-\hat{\phi})u (d_4 - d_5) 
		+ \bar{m}_{\text{in}} d_6 
		\right]
		\right\|^2 $}
\end{equation}

 A regularization term prevents overfitting:
\[
\mathcal{L}_{\text{total}} = \mathcal{L}_{\text{physics}} + \lambda_{\text{reg}} \cdot \mathcal{L}_{\text{reg}}.
\]

where:

\begin{equation*}
		\mathcal{L}_{\text{reg}} = \| \mathbf{W} \|^2
\end{equation*}

To minimize $\mathcal{L}_{\text{total}}$, the network weights $\mathbf{W}$ are updated via online gradient descent:
\begin{equation}
	\mathbf{W}_{k+1} = \mathbf{W}_k - \eta \nabla_{\mathbf{W}} \mathcal{L}_{\text{total}}
	\label{eq:weight_update}
\end{equation}
where:
\begin{itemize}
	\item $\eta > 0$ is the learning rate,
	\item $\nabla_{\mathbf{W}} \mathcal{L}_{\text{total}}$ denotes the gradient of the total loss with respect to weights,
	\item $k$ is the discrete-time update index.
\end{itemize}
The gradient $\nabla_{\mathbf{W}} \mathcal{L}_{\text{total}}$ is computed via backpropagation, with automatic differentiation used to evaluate the physics-based loss term $\mathcal{L}_{\text{physics}}$. This update occurs at each sampling instant during online operation.

This online gradient descent approach enables continuous adaptation of the PINN to evolving fault conditions. The physics-informed loss $\mathcal{L}_{\text{physics}}$ serves as a soft constraint, ensuring weight updates respect the underlying system dynamics. Computational efficiency is maintained through mini-batch processing of recent operational data within a sliding window, balancing real-time performance with estimation accuracy. The learning rate $\eta$ is tuned to ensure stable convergence while accommodating the slow thermal dynamics of HRSG systems.

The PINN respects known physical bounds by employing a sigmoid output activation function, ensuring \( \hat{\phi}(t) \in (0, 1) \). The slow dynamics of thermal systems support the use of smooth activation functions and continuous estimates. Also it runs online in parallel with the SMO and the controller. The estimated fault signal \( \hat{\phi}(t) \) is used to:
\begin{itemize}
	\item Adjust the observer dynamics for accurate state reconstruction,
	\item Adapt the control law for actuator-compensated temperature regulation.
\end{itemize}
This integration enables fault-tolerant performance even under unobservable fault conditions, preserving output regulation and system stability.

\subsection{Sliding Mode Control (SMC)}

To compensate for reduced actuation due to fault \( \phi(t) \), we employ a one-sided Sliding Mode Control law. The control input is defined as:
\begin{equation}
		\resizebox{\columnwidth}{!}{$
	u(t) = \frac{v(t)}{1 - \hat{\phi}(t) + \epsilon}, \quad 
	v(t) = 
	\begin{cases} 
		k \cdot \min\left(1, \dfrac{s(t)}{\delta}\right) & \text{if } s(t) > 0, \\
		0 & \text{otherwise},
	\end{cases} \quad 
	s(t) = x_1 - x_1^{\text{ref}}(t), \label{OneSidedController} $}
\end{equation}

where:
\begin{itemize}
	\item \( \hat{\phi}(t) \): fault estimate from the PINN,
	\item \( \epsilon > 0 \): regularization to prevent division by zero,
	\item \( k, \delta \): control gains,
	\item \( x_1^{\text{ref}}(t) \): desired outlet temperature trajectory.
\end{itemize}

This control strategy only activates when the outlet temperature exceeds the reference \( x_1 > x_1^{\text{ref}} \), injecting spray water (\( u > 0 \)) to reduce temperature. When the temperature is less than or equal to the reference, the controller remains idle (\( u = 0 \)). This \emph{one-sided control} minimizes unnecessary actuation while ensuring thermal safety by preventing overheating. Fig.~\ref{fig:smo_pinn_smc} summarizes the integration of observer, fault estimator, and controller components in the proposed fault-tolerant architecture.

\begin{figure}[h!]
	\centering
	\includegraphics[width=0.7\columnwidth]{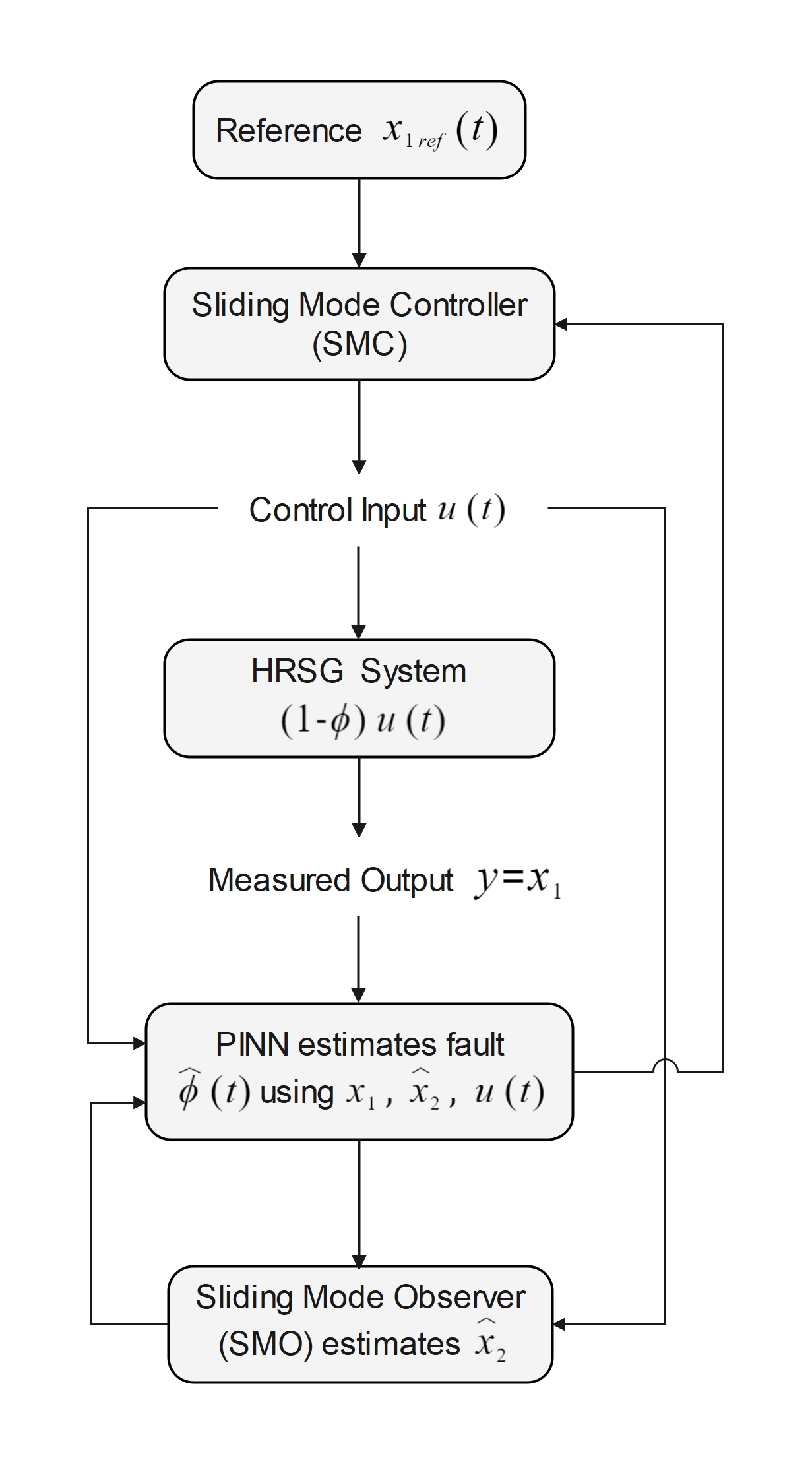}
	\caption{Schematic of the proposed SMO–PINN–SMC architecture. The Sliding Mode Observer (SMO) reconstructs unmeasured states; the Physics-Informed Neural Network (PINN) estimates the multiplicative actuator fault; and the Sliding Mode Controller (SMC) ensures fault-tolerant steam temperature regulation. The structure creates a closed-loop adaptive control scheme that responds to loss-of-effectiveness faults in real time.}
	\label{fig:smo_pinn_smc}
\end{figure}

\section{Stability}
\label{sec:stability}

We now show that, in the presence of actuator faults, the proposed SMO-PINN-SMC framework guarantees uniform ultimate boundedness of the state estimation and tracking errors.
\subsection{Closed-loop stability Analysis}
\begin{assumption}
	\label{assumption1}
	
	\textbf{Lipschitz Continuity\cite{khalil2002nonlinear}:} The nonlinear function \( f(x_2, u, \phi) = (d_2 + (1-\phi) u)(d_4 - x_2) - (1-\phi) u(d_4 - d_5) \) is assumed to be Lipschitz continuous in both the state variable \( x_2 \) and the fault parameter \( \phi \), with corresponding Lipschitz constants \( L_x > 0 \) and \( L_\phi > 0 \). That is,
	\[
	|f(x_2, u, \phi) - f(\hat{x}_2, u, \hat{\phi})| \leq L_x |x_2 - \hat{x}_2| + L_\phi |u| \cdot |\phi - \hat{\phi}|
	\]
	for all admissible \( x_2, \hat{x}_2 \in \mathbb{R} \), \( \phi, \hat{\phi} \in [0,1] \), and bounded control input \( u \in \mathbb{R} \).
	
	This assumption reflects the physical reality of the HRSG desuperheater system. The nonlinear function \( f(x_2, u, \phi) \), which models the effect of the water spray control input \( u \) (modulated by the actuator effectiveness \( \phi \)) on the intermediate temperature \( x_2 \), primarily involves bilinear terms of the form \( \phi u (d_4 - x_2) \). Physically, this represents the convective heat transfer between the spray and the steam, which varies smoothly with both the spray amount and the temperature difference.
	
	In practice, both the control input \( u \) and the actuator fault \( \phi \) are bounded due to physical constraints on the valve operation and system protection mechanisms. Moreover, temperature variations within the HRSG are naturally slow and continuous, as the thermal dynamics of the fluid and metal structures filter out high-frequency changes. This results in smooth changes in \( f(x_2, u, \phi) \) with respect to \( x_2 \) and \( \phi \).
\end{assumption}

\begin{assumption}
	\label{assumption2}
	
	The disturbance term \( K_1 d_1 - d_3 \), representing the mismatch between the expected and actual thermal effect on the superheater metal, is assumed to be bounded by a  constant \( M_b \).
	
	This assumption is thermodynamically justified because the fuel mass flow rate, which governs the gas turbine thermal energy output, is bounded by design and operational limits. The resulting hot gas flow transfers heat to the superheater metal through convection and radiation, characterized by bounded heat transfer coefficients. Crucially, the temperature evolution of the metallic structure is limited by its thermal inertia—specifically, its mass and heat capacity—which prevents arbitrarily fast changes. Consequently, the mismatch between the idealized, fuel-input-driven rate of temperature rise and the actual rate of change of the metal temperature remains bounded.

\end{assumption}
\begin{assumption}
	\label{assumption3}
	The parameter \(d_7\), representing the mass flow rate of steam leaving the desuperheater section, is positive and sufficiently small such that \(K_2 d_7^2 < 1\).
	
	This assumption ensures two key physical conditions. First, \(d_7 > 0\) guarantees continuous steam outflow, which is necessary for normal operation and thermal energy transfer in the HRSG. Second, the condition \(K_2 d_7^2 < 1\) prevents excessive nonlinear gain in the dynamic model, which could otherwise lead to unrealistic system behavior. Since \(K_2\) depends on the thermophysical properties and geometry of the desuperheater, and the steam outflow \(d_7\) is typically small in normal operating regimes, this condition is naturally satisfied in practical scenarios.
	
\end{assumption}

\begin{theorem}[Uniform Ultimate Boundedness]
	Consider the closed-loop HRSG system (1) with control law \eqref{OneSidedController}, Sliding Mode Observer \eqref{SMO}, and PINN estimator \eqref{PINN}. Under the Assumptions \ref{assumption1} to \ref{assumption3}, the errors \(e_1(t)\), \(e_2(t)\), and \(s(t)\) are uniformly ultimately bounded with observer gains satisfying:
	\begin{equation*}
		\lambda_1 > K_2d_7 + K_3L_x, \quad \lambda_2 > K_3L_\phi u_{\max}
	\end{equation*}
\end{theorem}

\begin{proof}
	
  Consider Lyapunov candidate
  \[
  V(\mathbf e) = \tfrac12 ( e_1^2 + e_2^2 + s^2 ),
  \qquad
  \mathbf e=[e_1,e_2,s]^\top.
  \]
  Here $V$ is positive definite and radially unbounded and $V=\tfrac12\|\mathbf e\|^2$.

  \medskip
  
From the SMO and error definitions \eqref{SMO} we have
  \begin{align}
  	\dot e_1 &= K_2 d_7 (e_2 - e_1) - \lambda_1 \,\mathrm{sat}\!\big(\tfrac{e_1}{\delta_1}\big), \label{e1dot}\\
  	\dot e_2 &= K_3\big(f(x_2,u,\phi)-f(\hat x_2,u,\hat\phi)\big) - \lambda_2 \,\mathrm{sat}\!\big(\tfrac{e_2}{\delta_2}\big). \label{e2dot}
  \end{align}
  
  Multiply \eqref{e1dot} by $e_1$ and use $|e_1e_2|\le |e_1||e_2|$:
  \begin{align}
  	e_1\dot e_1
  	&= K_2 d_7 e_1 e_2 - K_2 d_7 e_1^2 - \lambda_1 e_1 \,\mathrm{sat}\!\big(\tfrac{e_1}{\delta_1}\big) \nonumber\\
  	&\le K_2 d_7 |e_1||e_2| - K_2 d_7 e_1^2 - \lambda_1 |e_1| + \lambda_1\delta_1, \label{e1_bound}
  \end{align}

  For the second error, we apply Lipschitz continuity of $f$, therefore from \eqref{e2dot} we get
  \begin{align}
  	e_2\dot e_2
  	&\le K_3 L_x e_2^2 + K_3 L_\phi |u|\,|\phi-\hat\phi|\,|e_2|
  	- \lambda_2 e_2\,\mathrm{sat}\!\big(\tfrac{e_2}{\delta_2}\big) \nonumber\\
  	&\le K_3 L_x e_2^2 + K_3 L_\phi u_{\max}\Delta_{\phi}\,|e_2|
  	-\lambda_2 |e_2| + \lambda_2 \delta_2. \label{e2_intermediate_bound}
  \end{align}

  For the $s$ dynamics:
  \begin{align}
  	\dot{s} &= \dot{x}_1 - \dot{x}_1^{\mathrm{ref}},\nonumber \\ 
  	s \dot{s} &= s \left( \dot{x}_1 - \dot{x}_1^{\mathrm{ref}} \right) = s K_2 \left( K_1 d_1 + d_7 (x_2 - x_1) - d_3 \right) - s \dot{x}_1^{\mathrm{ref}} \nonumber\\
  	&= K_2 s \left( K_1 d_1 - d_3 + d_7 x_2 - d_7 (s + x_1^{\mathrm{ref}}) \right) - s \dot{x}_1^{\mathrm{ref}} \nonumber\\
  	&= -K_2 d_7 s^2 + K_2 s \left( K_1 d_1 - d_3 + d_7 x_2 - d_7 x_1^{\mathrm{ref}} \right) - s \dot{x}_1^{\mathrm{ref}} \nonumber\\
  	&\leq K_2 d_7 s^2 + |s| K_2 \left( M_b + d_7 M_{x_2} \right) + K_2 d_7 |x_1^{\mathrm{ref}}| |s| + |s| |\dot{x}_1^{\mathrm{ref}}|. \label{s_bound}
  \end{align}

  \medskip
  We now eliminate mixed absolute-value terms by standard Young inequalities.
  
  \emph{(i) The $e_1$--$e_2$ cross-term.} For any $\alpha>0$,
  \[
  K_2 d_7 |e_1||e_2|
  \le \frac{\alpha}{2} e_1^2 + \frac{(K_2 d_7)^2}{2\alpha} e_2^2.
  \]
  Substitute into \eqref{e1_bound} we obtain
  \begin{equation}
  	e_1\dot e_1
  	\le -\Big(K_2 d_7 - \frac{\alpha}{2}\Big)e_1^2
  	+ \frac{(K_2 d_7)^2}{2\alpha} e_2^2
  	-\lambda_1 |e_1| + \lambda_1\delta_1. \label{e1_after_young}
  \end{equation}
  
  \emph{(ii) The $e_2$ linear term from fault uncertainty.} For any $\mu>0$,
  \[
  K_3 L_\phi u_{\max}\Delta_{\phi}\,|e_2|
  \le \frac{\mu}{2} e_2^2 + \frac{(K_3 L_\phi u_{\max}\Delta_{\phi})^2}{2\mu}.
  \]
  Apply this in \eqref{e2_intermediate_bound}. Next, to convert the \(-\lambda_2|e_2|\) term into negative quadratic damping we use the inequality:
  \begin{equation}
  	-\lambda_2 |e_2| \le -\frac{\beta}{2} e_2^2 - \frac{\lambda_2^2}{2\beta}. \label{neg_lin_to_quad}
  \end{equation}
  Using these two bounds yields

 \begin{align}
 	e_2 \dot{e}_2 
 	&= e_2 \left( f(x_2,u,\phi) - f(\hat{x}_2,u,\hat{\phi}) \right) + \lambda_2 \operatorname{sat}\left(\frac{e_2}{\delta_2}\right) \nonumber\\
 	&\le e_2 \left( L_x |e_2| + L_\phi |u| |\phi - \hat{\phi}| \right) + \lambda_2 \operatorname{sat}\left(\frac{e_2}{\delta_2}\right) \nonumber\\
 	&= L_x e_2^2 + L_\phi |u| |e_2| |\phi - \hat{\phi}| + \lambda_2 \operatorname{sat}\left(\frac{e_2}{\delta_2}\right) \nonumber\\
 	&\le L_x e_2^2 + \frac{(L_\phi |u| |\phi - \hat{\phi}|)^2}{2 \mu} + \frac{\mu}{2} e_2^2 + \lambda_2 \operatorname{sat}\left(\frac{e_2}{\delta_2}\right) \nonumber\\
 	&\le \left( L_x + \frac{\mu}{2} + \frac{\beta}{2} \right) e_2^2 + \frac{(L_\phi |u| |\phi - \hat{\phi}|)^2}{2 \mu} + \lambda_2 \delta_2 - \frac{\lambda_2^2}{2 \beta}\nonumber \\
 	& \resizebox{\columnwidth}{!}{$\le \left( K_3 L_x + \frac{\mu}{2} - \frac{\beta}{2} + \frac{(K_2 d_7)^2}{2 \alpha} \right) e_2^2 + \frac{(K_3 L_\phi u_{\max} \Delta_\phi)^2}{2 \mu} + \lambda_2 \delta_2 - \frac{\lambda_2^2}{2 \beta}. \label{e2-afteryong} $}
 \end{align}

  \emph{(iii) The $s$-terms.} Use Young's inequalities with a free constant $c>0$ to get
  \begin{align*}
  	|s|K_2(M_b + d_7 M_{x_2}) &\le \frac{c}{2} s^2 + \frac{[K_2(M_b + d_7 M_{x_2})]^2}{2c},\label{s_young1}\\
  	K_2 d_7 |x_1^{\mathrm{ref}}||s| &\le \frac{1}{2} (x_1^{\mathrm{ref}})^2 + \frac{1}{2} K_2^2 d_7^2 s^2.
  \end{align*}
  Combine these with \eqref{s_bound} to obtain
  \begin{equation}
  	s\dot s \le \left(K_2 d_7 + \frac{c}{2} + \frac{1}{2}K_2^2 d_7^2\right)s^2
  	+ \frac{[K_2(M_b + d_7 M_{x_2})]^2}{2c} + \tfrac12 (x_1^{\mathrm{ref}})^2. \label{s_after_young}
  \end{equation}
  
  \medskip
  
  Add \eqref{e1_after_young}, \eqref{e2-afteryong} and \eqref{s_after_young}:

  \begin{align}
  	\dot V &= e_1\dot e_1 + e_2\dot e_2 + s\dot s \nonumber\\
  	&\le
  	-\Big(K_2 d_7 - \frac{\alpha}{2}\Big)e_1^2 \nonumber\\
  	&\quad + \Big( K_3 L_x + \frac{\mu}{2} - \frac{\beta}{2}
  	+ \frac{(K_2 d_7)^2}{2\alpha} \Big) e_2^2 \nonumber\\
  	&\quad + \Big( K_2 d_7 + \frac{c}{2} + \frac{1}{2}K_2^2 d_7^2 \Big) s^2
  	- \lambda_1 |e_1| \nonumber\\
  	& \resizebox{\columnwidth}{!}{$\quad + \underbrace{\lambda_1\delta_1
  		+ \frac{(K_3 L_\phi u_{\max}\Delta_{\phi})^2}{2\mu}
  		+ \lambda_2 \delta_2
  		+ \frac{[K_2(M_b + d_7 M_{x_2})]^2}{2c}
  		+ \tfrac12 (x_1^{\mathrm{ref}})^2
  		- \frac{\lambda_2^2}{2\beta}}_{\Theta(\alpha,\mu,\beta,c,\lambda_1,\lambda_2)}. $}
  	\label{Vdot_collected} 
  \end{align}

  \textbf{Choice of design parameters to obtain negativity:}
  We now choose the free parameters $\alpha,\mu,\beta,c>0$ and the observer gains $\lambda_1,\lambda_2$ so that the quadratic coefficients on $e_1^2,e_2^2,s^2$ are strictly negative.
  
  \begin{itemize}
  	\item Choose $\alpha$ small enough so that
  	\[
  	K_2 d_7 - \frac{\alpha}{2} > 0 \quad\Longrightarrow\quad
  	\text{coefficient of }e_1^2\;=\; -\Big(K_2 d_7 - \frac{\alpha}{2}\Big)\;<\;0.
  	\]
  	This is always possible for any $K_2d_7>0$ by picking $\alpha\in(0,2K_2 d_7)$.
  	
  	\item Choose $\beta>0$ large enough so that
  	\[
  	\frac{\beta}{2} > K_3 L_x + \frac{\mu}{2} + \frac{(K_2 d_7)^2}{2\alpha},
  	\]
  	which makes the coefficient of $e_2^2$ negative:
  	\[
  	K_3 L_x + \frac{\mu}{2} - \frac{\beta}{2} + \frac{(K_2 d_7)^2}{2\alpha} < 0.
  	\]
  	(Such a $\beta$ exists because the right-hand side is finite; increasing $\beta$ increases the left-hand side's negativity.)
  	
  	\item Choose $c>0$ sufficiently small so that
  	\[
  	K_2 d_7 + \frac{c}{2} + \frac{1}{2}K_2^2 d_7^2 <: \bar\kappa_s
  	\]
  	is a finite positive number. By taking $c$ small we make this coefficient as close to $K_2 d_7 + \tfrac12 K_2^2 d_7^2$ as needed.
  \end{itemize}
  
  With $\alpha,\mu,\beta,c$ chosen as above, there exist positive parameters
  \[
  \kappa_1 := \frac{1}{2}\Big(K_2 d_7 - \frac{\alpha}{2}\Big) > 0, \qquad
  \kappa_2 := \frac{1}{2}\Big( \frac{\beta}{2} - K_3 L_x - \frac{\mu}{2} - \frac{(K_2 d_7)^2}{2\alpha}\Big) > 0,
  \]
  and
  \[
  \kappa_3 := \frac{1}{2}\Big( K_2 d_7 + \frac{c}{2} + \frac{1}{2}K_2^2 d_7^2\Big) > 0,
  \]
  such that the quadratic part of \eqref{Vdot_collected} can be lower bounded by
  \[
  -2\kappa_1 e_1^2 - 2\kappa_2 e_2^2 - 2\kappa_3 s^2.
  \]
  
  The term $-\lambda_1|e_1|$ further helps negativity for large $|e_1|$, to avoid technicalities we simply absorb this linear term into the quadratic negative definite part by noting that for any $\eta>0$,
  \[
  -\lambda_1|e_1| \le -\eta e_1^2 + \frac{\lambda_1^2}{4\eta}.
  \]
  Choosing $\eta \in(0,\kappa_1)$ we keep a net negative quadratic coefficient on $e_1^2$. Concretely, pick $\lambda_1$ large enough that the constant $\lambda_1^2/(4\eta)$ is acceptable in the ultimate bound.
  
  \medskip
  
  Putting these estimates together, we obtain from \eqref{Vdot_collected} the differential inequality
  \begin{equation}
  	\dot V \le -2\kappa V + \widetilde{\Theta},
  	\label{Vdot_final}
  \end{equation}
  where
  \[
  \kappa := \min(\kappa_1-\eta,\;\kappa_2,\;\kappa_3) > 0
  \]
  (using the small $\eta$ chosen above) and $\widetilde{\Theta}>0$ is a finite parameter that collects all remaining terms:
  \[\resizebox{\columnwidth}{!}{$
  \widetilde{\Theta} =
  \lambda_1\delta_1
  + \frac{(K_3 L_\phi u_{\max}\Delta_{\phi})^2}{2\mu}
  + \lambda_2 \delta_2
  + \frac{[K_2(M_b + d_7 M_{x_2})]^2}{2c}
  + \tfrac12 (x_1^{\mathrm{ref}})^2
  + \frac{\lambda_1^2}{4\eta}
  + \frac{\lambda_2^2}{2\beta_1}, $}
  \]
  where $\beta_1$ is a convenient constant coming from \eqref{neg_lin_to_quad} bookkeeping (absorbed into the earlier choice of $\beta$).
  
  \medskip
  
 From \eqref{Vdot_final}, the comparison scalar ODE
  \[
  \dot W = -2\kappa W + \widetilde{\Theta}, \qquad W(0)=V(0)
  \]
  has the solution
  \[
  W(t) = e^{-2\kappa t}V(0) + \frac{\widetilde{\Theta}}{2\kappa}\big(1-e^{-2\kappa t}\big).
  \]
  By the comparison lemma $V(t)\le W(t)$ for all $t\ge0$, and therefore
  \[
  \limsup_{t\to\infty} V(t) \le \frac{\widetilde{\Theta}}{2\kappa}.
  \]
  Returning to the Euclidean norm $\|\mathbf e\|^2 = 2V$ we obtain the uniform ultimate bound
  \[
  \limsup_{t\to\infty} \|\mathbf e(t)\|
  \le \sqrt{\frac{\widetilde{\Theta}}{\kappa}}.
  \]
  Component-wise bounds follow immediately: $|e_i(t)| \le \|\mathbf e(t)\|$.
  
  \medskip

  With these choices, the closed-loop estimation error $\mathbf e$ is uniformly ultimately bounded. This completes the proof.

\end{proof}

\subsection{Boundedness of PINN Weights}

\begin{assumption}
	\label{assumption4}

	The total PINN loss function $L_{\mathrm{total}}(W)$ has a gradient 
	$\nabla_W L_{\mathrm{total}}(W)$ that is $L_g$-Lipschitz continuous on a compact neighborhood $\Omega$ of the minimizer $W^\ast$:
	\[
	\|\nabla_W L_{\mathrm{total}}(W_1) - \nabla_W L_{\mathrm{total}}(W_2)\| 
	\le L_g \|W_1 - W_2\|,
	\quad \forall W_1,W_2 \in \Omega.
	\]
	Moreover, $L_{\mathrm{total}}$ is convex on $\Omega$:
	\[
	\langle \nabla_W L_{\mathrm{total}}(W) - \nabla_W L_{\mathrm{total}}(W^\ast),\, W - W^\ast \rangle 
	\ge \mu_g \|W - W^\ast\|^2, 
	\quad \forall W \in \Omega,
	\]
	for some constants $\mu_g>0$ and $L_g>0$.
\end{assumption}

\begin{lemma}
	\label{lemma1}

	Under Assumption~\ref{assumption4}, for any step size $0 < \eta < \frac{2}{L_g}$, the following contraction inequality holds:
	\[
	\| W - W^\ast - \eta \big( \nabla_W L_{\mathrm{total}}(W) - \nabla_W L_{\mathrm{total}}(W^\ast) \big) \|
	\le \rho \, \|W - W^\ast\|,
	\]
	where
	\[
	\rho := \sqrt{\,1 - 2\eta\mu_g + \eta^2 L_g^2\,} \quad \text{satisfies} \quad 0 < \rho < 1.
	\]
\end{lemma}

\begin{proof}
	Expanding the squared norm:
	\[
	\begin{aligned}
		&\| W - W^\ast - \eta (\nabla_W L_{\mathrm{total}}(W) - \nabla_W L_{\mathrm{total}}(W^\ast)) \|^2 \\
		&= \|W - W^\ast\|^2 
		- 2\eta \langle \nabla_W L_{\mathrm{total}}(W) - \nabla_W L_{\mathrm{total}}(W^\ast),\, W - W^\ast \rangle \\
		&\quad + \eta^2 \|\nabla_W L_{\mathrm{total}}(W) - \nabla_W L_{\mathrm{total}}(W^\ast)\|^2.
	\end{aligned}
	\]
	By convexity and $L_g$-smoothness:
	\[
	\langle \nabla_W L_{\mathrm{total}}(W) - \nabla_W L_{\mathrm{total}}(W^\ast),\, W - W^\ast \rangle 
	\ge \mu_g \|W - W^\ast\|^2,
	\]
	\[
	\|\nabla_W L_{\mathrm{total}}(W) - \nabla_W L_{\mathrm{total}}(W^\ast)\| 
	\le L_g \|W - W^\ast\|.
	\]
	Substituting these bounds into the expansion yields:
	\[
	\begin{aligned}
		&\| W - W^\ast - \eta ( \nabla_W L_{\mathrm{total}}(W) - \nabla_W L_{\mathrm{total}}(W^\ast) ) \|^2 \\
		&\quad \le \| W - W^\ast \|^2 - 2\eta \mu_g \| W - W^\ast \|^2 + \eta^2 L_g^2 \| W - W^\ast \|^2 \\
		&\quad = \left( 1 - 2\eta\mu_g + \eta^2 L_g^2 \right) \| W - W^\ast \|^2.
	\end{aligned}
	\]
	Taking square roots gives the contraction inequality:
	\[
	\| W - W^\ast - \eta ( \nabla_W L_{\mathrm{total}}(W) - \nabla_W L_{\mathrm{total}}(W^\ast) ) \| 
	\le \rho \| W - W^\ast \|,
	\]
	with $\rho = \sqrt{1 - 2\eta\mu_g + \eta^2 L_g^2} < 1$.
\end{proof}

\begin{theorem}
	\label{theorem2}
	\textbf{(Bounded Weight Deviation).}
	Let $\widetilde{W}_k = W_k - W^\ast$ denote the weight deviation at iteration $k$. 
	Suppose Assumption~\ref{assumption4} holds on a neighborhood $\Omega$ of $W^\ast$, the iterates satisfy $W_k \in \Omega$ for all $k$, the ideal weights satisfy $\|W^\ast\| \le M$, and the gradient mismatch at $W^\ast$ is bounded by $\|\nabla_W L_{\mathrm{total}}(W^\ast)\| \le \zeta$.
	If $0 < \eta < \tfrac{2\mu_g}{L_g^2}$, then the PINN weight deviation obeys:
	\[
	\|\widetilde{W}_k\| \le \rho^k \|\widetilde{W}_0\| + \frac{\eta\zeta}{1 - \rho}, 
	\qquad \rho := \sqrt{1 - 2\eta\mu_g + \eta^2 L_g^2} \in (0,1).
	\]
	In particular, $\{\widetilde{W}_k\}$ is uniformly bounded for all $k \ge 0$, and the bound decays geometrically to a neighborhood of radius $\frac{\eta\zeta}{1 - \rho}$.
\end{theorem}

\begin{proof}
	From the weight update law,
	\[
	\widetilde{W}_{k+1} = \widetilde{W}_k - \eta\big(\nabla_W L_{\mathrm{total}}(W_k) - \nabla_W L_{\mathrm{total}}(W^\ast)\big) - \eta\, \nabla_W L_{\mathrm{total}}(W^\ast),
	\]
	take norms and apply the triangle inequality together with Lemma~\ref{lemma1} (valid on $\Omega$):
	\[ \resizebox{\columnwidth}{!}{$
	\|\widetilde{W}_{k+1}\| 
	\le \big\| \widetilde{W}_k - \eta\big(\nabla_W L_{\mathrm{total}}(W_k) - \nabla_W L_{\mathrm{total}}(W^\ast)\big) \big\|
	+ \eta \|\nabla_W L_{\mathrm{total}}(W^\ast)\|
	\le \rho\, \|\widetilde{W}_k\| + \eta\zeta. $}
	\]

	Solving the linear recurrence gives
	\[
	\|\widetilde{W}_k\| \le \rho^k \|\widetilde{W}_0\| + \eta\zeta \sum_{i=0}^{k-1} \rho^i 
	= \rho^k \|\widetilde{W}_0\| + \frac{\eta\zeta\,(1 - \rho^k)}{1 - \rho}
	\le \rho^k \|\widetilde{W}_0\| + \frac{\eta\zeta}{1 - \rho}.
	\]
	Finally, $\rho<1$ holds because $0<\eta<\tfrac{2\mu_g}{L_g^2}$ implies $1 - 2\eta\mu_g + \eta^2 L_g^2 < 1$.

\end{proof}

\begin{corollary}
	The actual weights remain bounded:
	\[
	\|W_k\| \le M + \rho^k \|\widetilde{W}_{0}\| + \frac{\eta \zeta}{1 - \rho}
	\qquad \forall\, k \ge 0.
	\]
\end{corollary}

\section{Simulation Results}
\label{sec:results}

To evaluate the effectiveness of the proposed SMO-PINN-SMC framework, we conducted simulation studies using real operational data collected from the Parehsar combined cycle power plant. The objective is to validate the integrated performance of the PINN for fault estimation, the SMO for state reconstruction, and the SMC for robust tracking under actuator degradation.

All dynamic models, estimators, and controllers were implemented in
Python and tested directly on measurement data from a functioning HRSG
superheater unit. The results highlight key aspects of the system performance
under fault-tolerant control.

\subsection{Validation Using Continuous Maintenance Data}
\label{subsec:maintenance_validation}

To validate the fault estimation performance, we leverage continuous valve position measurements from maintenance logs at Pareh-sar power plant during a documented 1-hour valve degradation event. Using high-frequency position sensors, we calculate the true fault level in real-time:

\begin{equation*}
	\phi_{\text{true}}(t) = 1 - \frac{u_{\text{actual}}(t)}{u_{\text{cmd}}(t)}
	\label{eq:true_fault}
\end{equation*}

where:
\begin{itemize}
	\item $u_{\text{cmd}}(t)$: Valve position commanded by DCS 
	\item $u_{\text{actual}}(t)$: Actual valve position measured by maintenance sensors 
\end{itemize}

Fig. \ref{fig:phi_hat} shows the direct comparison between the PINN estimate $\hat{\phi}(t)$ and the continuously measured $\phi_{\text{true}}(t)$.

The key Observations are:
\begin{enumerate}
	\item \textit{Transient Tracking}: The PINN accurately captures rapid fault development.
	\item \textit{Steady-State Accuracy}: During stable periods, estimation error remains small.
	\item \textit{Noise Rejection}: High-frequency sensor noise in $u_{\text{actual}}$ is effectively filtered while preserving fault dynamics
	\item \textit{Control Impact}: The controller maintained the steam temperature within $\pm 1^{\circ}\text{C}$ despite $\phi$ changes.
\end{enumerate}

This continuous validation demonstrates the PINN's capability to quantify actuator degradation with precision sufficient for:
\begin{itemize}
	\item Automated fault compensation without operator intervention
	\item Real-time health monitoring of critical valves
\end{itemize}

Fig.~\ref{fig:phi_hat} shows the estimated actuator fault \( \hat{\phi}(t) \) produced by the PINN. The estimation remains smooth and bounded throughout the operation and accurately captures the time-varying nature of the fault. By incorporating the governing physics into its loss function, the PINN maintains dynamic consistency while providing real-time estimates of the actuator effectiveness.

Fig.~\ref{fig:SMO} presents the state estimation results for the unmeasured variable \( x_2(t) \), representing the steam temperature within the attemperator section. Despite the absence of direct measurements, the SMO effectively reconstructs \( \hat{x}_2(t) \), which is essential for both control synthesis and fault inference.

The accuracy of this state reconstruction is quantified in Fig.~\ref{fig:est_error}, which displays the absolute estimation error \( |x_2(t) - \hat{x}_2(t)| \). The error converges rapidly to a small neighborhood around zero and remains bounded throughout, even under active control and varying fault conditions. This behavior aligns with the theoretical stability guarantees established in the previous section.

Fig.~\ref{fig:SMC_output} compares the system output temperature \( x_1(t) \) with the desired reference trajectory \( x_1^{\text{ref}}(t) \), and also includes the baseline PID response for comparison. The proposed SMC ensures tracking after a brief transient, with bounded fluctuations that reflect the expected chattering behavior of sliding mode control. The results confirm the controller's resilience to fault-induced degradation.

The tracking performance is further highlighted in Fig.~\ref{fig:tracking_error}, which shows the tracking error \( e(t) = x_1(t) - x_1^{\text{ref}}(t) \). The error decreases sharply and remains confined within a narrow boundary, validating the stability and precision of the closed-loop system.

Fig.~\ref{fig:control_input} compares the control inputs $u(t)$ from the proposed SMC-based controller and the conventional PID controller. The PID controller parameters used for the comparison were tuned based on the HRSG superheater dynamics: proportional gain \(K_p = 1.5\), integral gain \(K_i = 96\), and derivative gain \(K_d = 0\) (derivative action was not used). These values were selected to balance responsiveness and stability under normal operating conditions.
 The SMC input exhibits sharp, event-triggered activations in response to the fault, resulting in a more aggressive valve opening to compensate for the loss of effectiveness. In contrast, the PID controller reacts slowly, treating the actuator as healthy and failing to adjust adequately to the fault. Since the SMC strategy is fault-tolerant, it intentionally increases the control input to open the valve further and counteract the reduced effectiveness of the actuator. This difference in control behavior leads to superior fault compensation under SMC, allowing the steam temperature to fluctuate around the desired value. In contrast, the PID case shows noticeable overshoot and a prolonged settling time before reaching the setpoint.

Collectively, these results demonstrate that the proposed architecture achieves reliable fault estimation, accurate state reconstruction, and robust output tracking under partial actuator failure. The combination of PINN-based learning and model-based control enables real-time adaptability without requiring full state measurement or fault instrumentation.

\begin{figure}[h!]
	\centering
	\includegraphics[width=0.8\columnwidth]{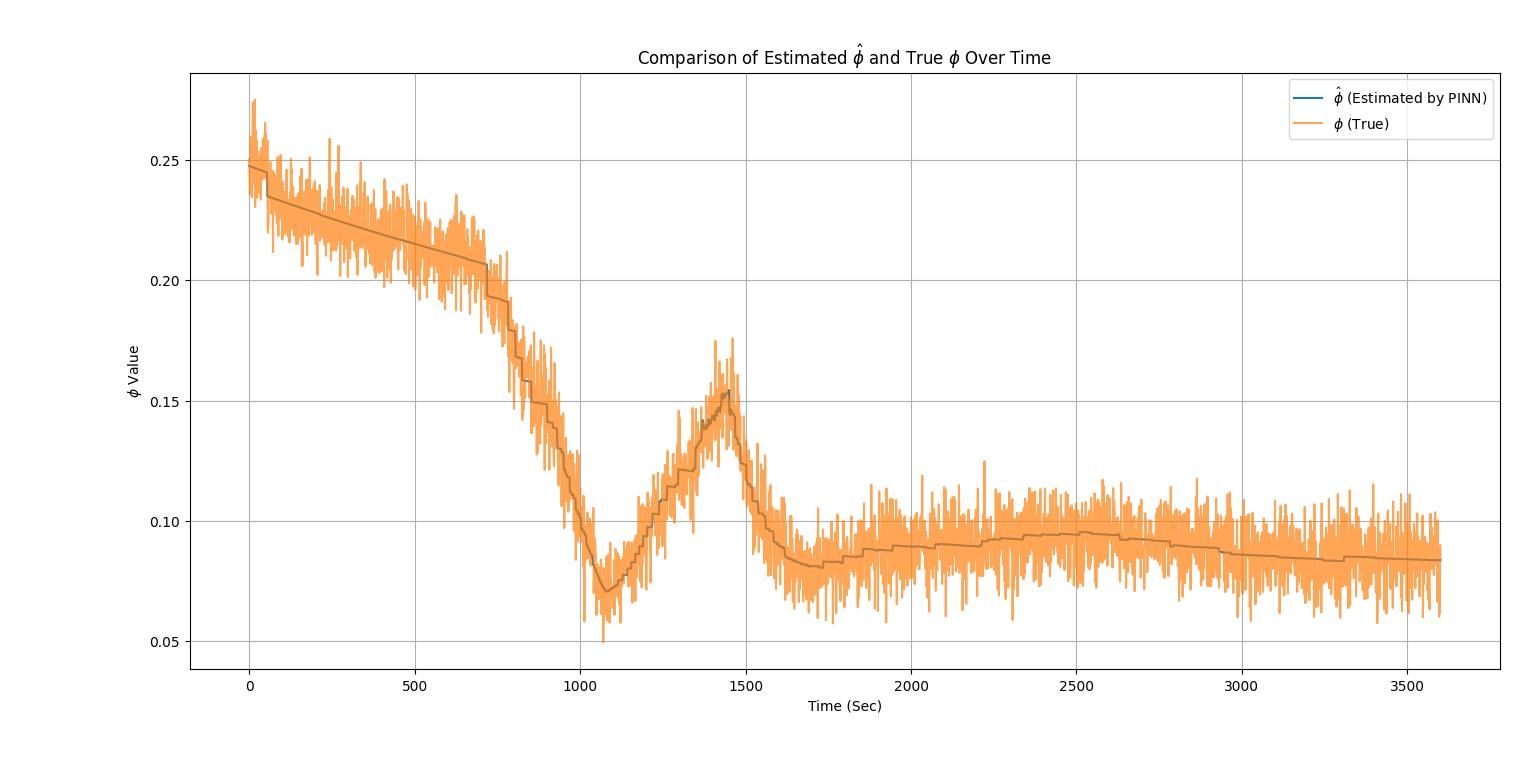}
	\caption{PINN Estimated Fault \(\hat{\phi}(t)\)}
	\label{fig:phi_hat}
\end{figure}

\begin{figure}[h!]
	\centering
	\includegraphics[width=0.8\columnwidth]{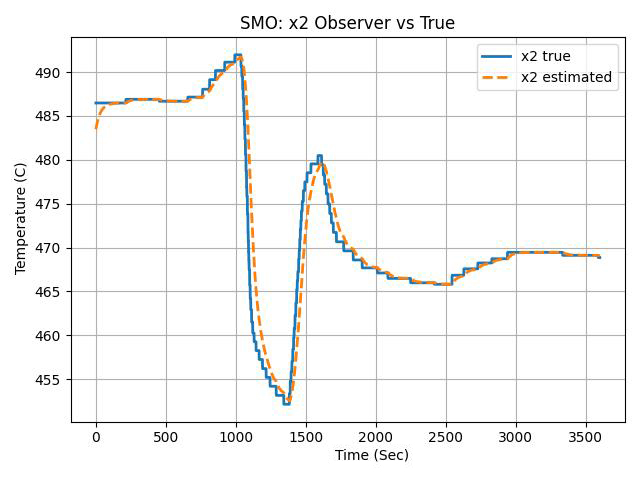}
	\caption{SMO: Estimated State \(\hat{x}_2(t)\) vs real \(x_2(t)\)}
	\label{fig:SMO}
\end{figure}

\begin{figure}[h!]
	\centering
	\includegraphics[width=0.8\columnwidth]{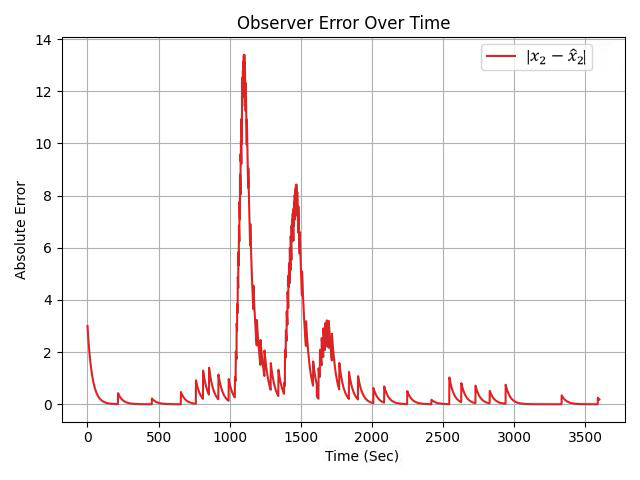}
	\caption{Observer Estimation Error \( |x_2 - \hat{x}_2| \)}
	\label{fig:est_error}
\end{figure}

\begin{figure}[h!]
	\centering
	\includegraphics[width=0.8\columnwidth]{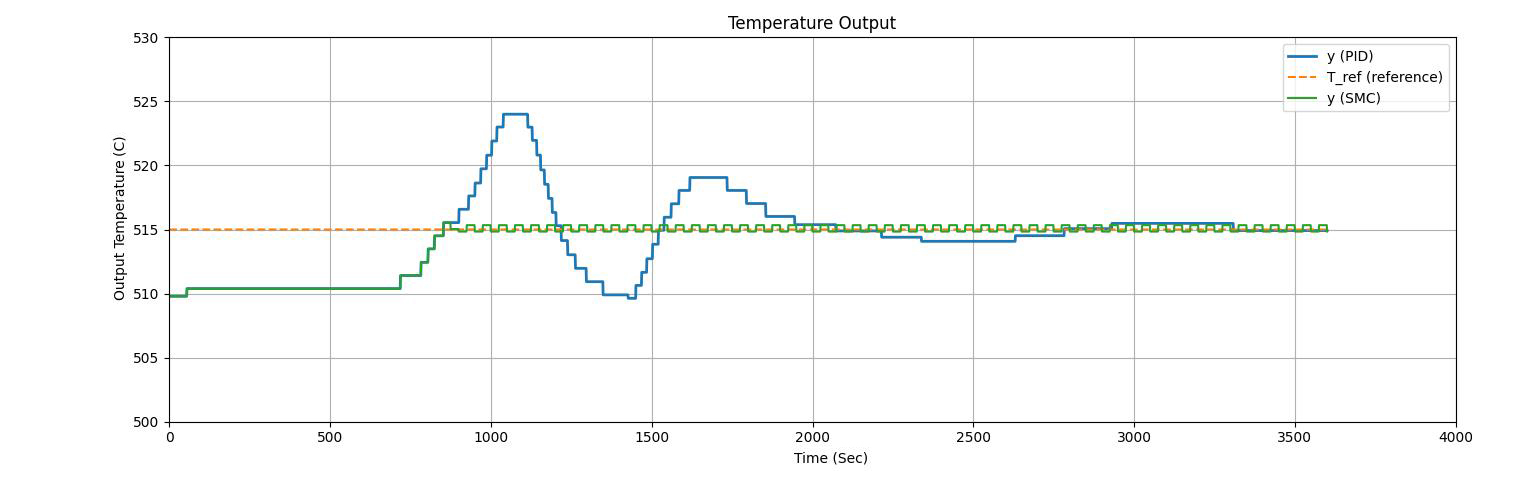}
	\caption{Controlled Output Temperature via SMC}
	\label{fig:SMC_output}
\end{figure}

\begin{figure}[h!]
	\centering
	\includegraphics[width=0.8\columnwidth]{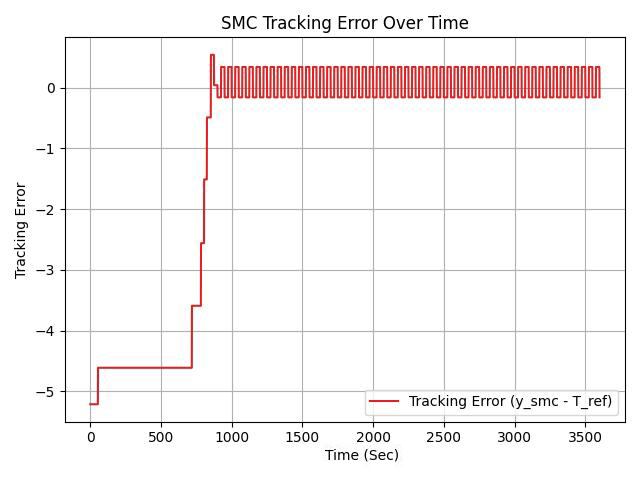}
	\caption{Tracking Error \( e(t) = x_1 - x_1^{\text{ref}} \)}
	\label{fig:tracking_error}
\end{figure}

\begin{figure}[h!]
	\centering
	\includegraphics[width=1\columnwidth]{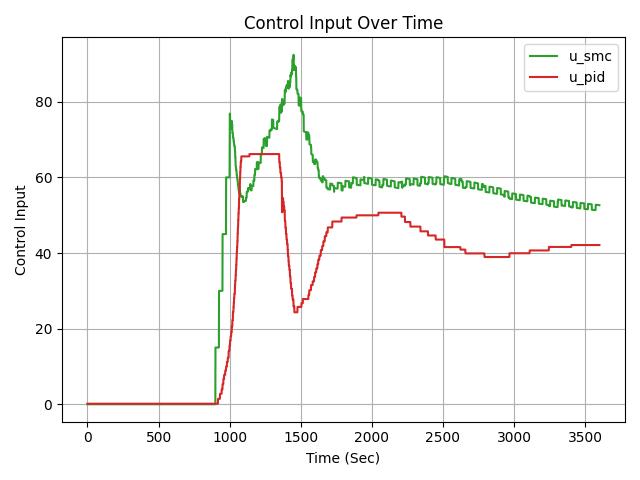}
	\caption{Control Input \( u(t) \) from the SMC compared to conventional PID}
	\label{fig:control_input}
\end{figure}

\section{Conclusion}
\label{sec:conclusion}

This study proposed a fault-tolerant control framework for regulating superheater steam temperature in heat recovery steam generator systems subject to multiplicative actuator faults. The architecture integrates a sliding mode observer and a physics-informed neural network to achieve simultaneous estimation of the unmeasured internal state and time-varying actuator effectiveness, enabling real-time control adaptation without direct fault measurements.

The actuator fault—modeled as a dynamic loss of spray valve effectiveness—was addressed using the PINN, which leveraged the system's nonlinear dynamics within a physics-informed loss function to estimate the fault signal. The SMO provided state estimation despite measurement constraints, and the combined SMO-PINN structure enabled accurate reconstruction of latent system variables and actuator behavior.

A Lyapunov-based analysis provided the uniform ultimate boundedness of the estimation and tracking errors under reasonable assumptions. The theoretical results were evaluated using simulations based on real HRSG operational data from the Parehsar power plant. The proposed controller achieved smooth input signals, robust tracking performance, and effective compensation under actuator degradation. Additionally, the modular proposed architecture is readily extendable to other nonlinear industrial systems with similar fault structures.

In summary, the SMO-PINN-SMC framework offers a practical and scalable solution for fault-tolerant temperature regulation in HRSG systems. Future research directions include handling sensor faults and deploying the method for predictive maintenance and early fault diagnosis in combined-cycle power plants.

\bibliographystyle{elsarticle-num}
\bibliography{mainrefs}

\begin{thebibliography}{10}
\expandafter\ifx\csname url\endcsname\relax
  \def\url#1{\texttt{#1}}\fi
\expandafter\ifx\csname urlprefix\endcsname\relax\def\urlprefix{URL }\fi
\expandafter\ifx\csname href\endcsname\relax
  \def\href#1#2{#2} \def\path#1{#1}\fi

\bibitem{latif2023failure}
N.~Latif, T.~Triwibowo, H.~Yuliani, V.~Astini, Failure investigation of
  superheater through investigate the nearest component, in: E3S Web of
  Conferences, Vol. 430, EDP Sciences, 2023, p. 01242.

\bibitem{sanaye2007transient}
S.~Sanaye, M.~Rezazadeh, Transient thermal modelling of heat recovery steam
  generators in combined cycle power plants, International journal of energy
  research 31~(11) (2007) 1047--1063.

\bibitem{self2018effects}
S.~J. Self, M.~A. Rosen, B.~V. Reddy, Effects of oxy-fuel combustion on
  performance of heat recovery steam generators, European Journal of
  Sustainable Development Research 2~(2) (2018) 22.

\bibitem{mcconnell2023modeling}
J.~McConnell, T.~Das, A.~Caesar, P.~Veeravalli, Modeling and simulation of a
  multistage heat recovery steam generator, Simulation 99~(2) (2023) 169--182.

\bibitem{sun2014study}
D.~D. Sun, C.~Yang, F.~Zeng, Study on gas turbine-based cchp system with
  multi-objective evaluation index, Advanced Materials Research 860 (2014)
  1366--1369.

\bibitem{mcconnell2019multi}
J.~McConnell, T.~Das, A.~Caesar, J.~Hoy, P.~Veeravalli, Multi-physics dynamic
  modeling and transient simulation of a multi-stage heat recovery steam
  generator (hrsg), in: Dynamic Systems and Control Conference, Vol. 59155,
  American Society of Mechanical Engineers, 2019, p. V002T13A002.

\bibitem{lv2019dependency}
X.~Lv, D.~Zhou, L.~Ma, Y.~Tang, Dependency model-based multiple fault diagnosis
  using knowledge of test result and fault prior probability, Applied Sciences
  9~(2) (2019) 311.

\bibitem{zhang2020distributed}
T.~Zhang, W.~Zhang, Q.~Zhao, Y.~Du, J.~Chen, J.~Zhao, Distributed real-time
  state estimation for combined heat and power systems, Journal of Modern Power
  Systems and Clean Energy 9~(2) (2020) 316--327.

\bibitem{mukhopadhyay2011failure}
G.~Mukhopadhyay, S.~Bhattacharyya, Failure analysis of an attemperator in a
  steam line of a boiler, Engineering Failure Analysis 18~(5) (2011)
  1359--1365.

\bibitem{shokouhmand2015failure}
H.~Shokouhmand, B.~Ghadimi, R.~Espanani, Failure analysis and retrofitting of
  superheater tubes in utility boiler, Engineering Failure Analysis 50 (2015)
  20--28.

\bibitem{nurbanasari2023failure}
M.~Nurbanasari, B.~Pramono, I.~G. Darmadi, I.~T. Saputro, G.~Gumilar,
  A.~Ekajati, et~al., Failure study of secondary superheater tube out header
  damage in a 600-mw coal power plant, Engineering Failure Analysis 150 (2023)
  107349.

\bibitem{kochmanski2024failure}
P.~Kochma{\'n}ski, S.~Fryska, A.~E. Kochma{\'n}ska, Failure analysis of steam
  superheater boiler tube made of astm t22 steel, Engineering Failure Analysis
  162 (2024) 108366.

\bibitem{gibbons_how_2025}
J.~Gibbons, K.~Mathews, How {Desuperheater} {Nozzle} {Testing} {Prevents} or
  {Predicts} {Failures}, automation.com, A subsidiary of the International
  Society of Automation (ISA) 4 (May 2025).

\bibitem{abdin2022state}
Z.~U. Abdin, A.~Rachid, State observer for water-based hybrid pv/t system with
  unknown input, in: IECON 2022--48th Annual Conference of the IEEE Industrial
  Electronics Society, IEEE, 2022, pp. 1--6.

\bibitem{inyang2021health}
U.~Inyang, I.~Petrunin, I.~Jennions, Health condition estimation of bearings
  with multiple faults by a composite learning-based approach, Sensors 21~(13)
  (2021) 4424.

\bibitem{zhang2024sliding}
R.~Zhang, P.~Li, W.~Liang, Sliding-mode observer-based fault diagnosis and
  fault-tolerant control of the main drive system of rolling mill, Transactions
  of the Institute of Measurement and Control 46~(7) (2024) 1274--1282.

\bibitem{tang2023data}
W.~Tang, Data-driven state observation for nonlinear systems based on online
  learning, AIChE Journal 69~(12) (2023) e18224.

\bibitem{zhou2023virtual}
H.~Zhou, Virtual reference setting and its application in thermal power unit
  control, in: Journal of Physics: Conference Series, Vol. 2528, IOP
  Publishing, 2023, p. 012071.

\bibitem{li2017superheat}
N.~Li, X.~Wang, Superheat control of evaporator outlet in air conditioning, in:
  2017 2nd Joint International Information Technology, Mechanical and
  Electronic Engineering Conference (JIMEC 2017), Atlantis Press, 2017, pp.
  391--395.

\bibitem{chen2022phase}
Z.~Chen, Y.-S. Hao, L.~Sun, Z.-g. Su, Phase compensation based active
  disturbance rejection control for high order superheated steam temperature
  system, Control Engineering Practice 126 (2022) 105200.

\bibitem{wu2019superheated}
Z.~Wu, T.~He, D.~Li, Y.~Xue, L.~Sun, L.~Sun, Superheated steam temperature
  control based on modified active disturbance rejection control, Control
  Engineering Practice 83 (2019) 83--97.

\bibitem{alamoodi2017nonlinear}
N.~Alamoodi, P.~Daoutidis, Nonlinear control of coal-fired steam power plants,
  Control Engineering Practice 60 (2017) 63--75.

\bibitem{wang2018fixed}
Z.~Wang, Y.~Bai, J.~Xie, Z.~Li, C.~Ma, P.~Liu, Y.~Zhang, Fixed-time
  sliding-mode fault-tolerant control of waste heat power generator systems,
  Complexity 2018~(1) (2018) 3580628.

\bibitem{wang2019evaluation}
Y.~Wang, D.~Bhattacharyya, R.~Turton, Evaluation of novel configurations of
  natural gas combined cycle (ngcc) power plants for load-following operation
  using dynamic modeling and optimization, Energy \& Fuels 34~(1) (2019)
  1053--1070.

\bibitem{vcehil2017novel}
M.~{\v{C}}ehil, S.~Katuli{\'c}, D.~R. Schneider, Novel method for determining
  optimal heat-exchanger layout for heat recovery steam generators, Energy
  conversion and management 149 (2017) 851--859.

\bibitem{guo2023nonsingular}
X.~Guo, X.~Liao, Q.~Wang, Y.~Liu, Nonsingular terminal sliding mode control
  with sliding perturbation observer for a permanent-magnet spherical actuator,
  Proceedings of the Institution of Mechanical Engineers, Part I: Journal of
  Systems and Control Engineering 237~(2) (2023) 259--271.

\bibitem{li2023physics}
H.~Li, L.~Gou, H.~Li, Z.~Liu, Physics-guided neural network model for
  aeroengine control system sensor fault diagnosis under dynamic conditions,
  Aerospace 10~(7) (2023) 644.

\bibitem{retzler2024learning}
A.~Retzler, R.~T{\'o}th, M.~Schoukens, G.~I. Beintema, J.~Weigand, J.-P.
  No{\"e}l, Z.~Koll{\'a}r, J.~Swevers, Learning-based augmentation of
  physics-based models: an industrial robot use case, Data-Centric Engineering
  5 (2024) e12.

\bibitem{jenkins2018convergence}
B.~M. Jenkins, A.~M. Annaswamy, E.~Lavretsky, T.~E. Gibson, Convergence
  properties of adaptive systems and the definition of exponential stability,
  SIAM journal on control and optimization 56~(4) (2018) 2463--2484.

\bibitem{bolderman2024physics}
M.~Bolderman, H.~Butler, S.~Koekebakker, E.~Van~Horssen, R.~Kamidi,
  T.~Spaan-Burke, N.~Strijbosch, M.~Lazar, Physics-guided neural networks for
  feedforward control with input-to-state-stability guarantees, Control
  Engineering Practice 145 (2024) 105851.

\bibitem{van2020noisy}
H.~J. Van~Waarde, M.~K. Camlibel, M.~Mesbahi, From noisy data to feedback
  controllers: Nonconservative design via a matrix s-lemma, IEEE Transactions
  on Automatic Control 67~(1) (2020) 162--175.

\bibitem{sabounchi2021fltrl}
M.~Sabounchi, J.~Wei-Kocsis, Fltrl: a fuzzy-logic transfer learning powered
  reinforcement learning method for intelligent online control in power
  systems, in: North American Fuzzy Information Processing Society Annual
  Conference, Springer, 2021, pp. 368--379.

\bibitem{dijoux2022experimental}
E.~Dijoux, N.~Y. Steiner, M.~Benne, M.-C. Pera, B.~Grondin-Perez, Experimental
  validation of an active fault tolerant control strategy applied to a proton
  exchange membrane fuel cell, Electrochem 3~(4) (2022) 633--652.

\bibitem{adumene2015performance}
S.~Adumene, B.~T. Lebele-Alawa, et~al., Performance optimization of dual
  pressure heat recovery steam generator (hrsg) in the tropical rainforest,
  Engineering 7~(06) (2015) 347.

\bibitem{kaviri2012modeling}
G.~Kaviri, M.~Mohd~Jafar, M.~Tholudin, Modeling and optimization of heat
  recovery heat exchanger, Applied Mechanics and Materials 110 (2012)
  2448--2452.

\bibitem{fanoodi2025pinn}
M.~Fanoodi, F.~Abdollahi, M.~A. Shoorehdeli, M.~Maboodi,
  \href{http://arxiv.org/abs/2512.00990}{Fault-{Tolerant} {Temperature}
  {Control} of {HRSG} {Superheaters}: {Stability} {Analysis} {Under} {Valve}
  {Leakage} {Using} {Physics}-{Informed} {Neural} {Networks}}, arXiv:2512.00990
  [eess] (Nov. 2025).
\newblock \href {https://doi.org/10.48550/arXiv.2512.00990}
  {\path{doi:10.48550/arXiv.2512.00990}}.
\newline\urlprefix\url{http://arxiv.org/abs/2512.00990}

\bibitem{khalil2002nonlinear}
H.~K. Khalil, J.~W. Grizzle, Nonlinear systems, Vol.~3, Prentice hall Upper
  Saddle River, NJ, 2002.

\end{thebibliography}

\end{document}